\documentclass[journal]{IEEEtran}

\usepackage{algorithm}
\usepackage{algorithmic}
\usepackage{ifpdf}
\usepackage{latexsym,amsfonts,amssymb,amsmath,graphicx,epsf,cite,bbm,float}
\usepackage{ifpdf,flexisym}
\usepackage{epstopdf}
\usepackage{caption}
\usepackage{subcaption}
\usepackage{comment}
\usepackage{amsthm}
\usepackage{mathtools}
\usepackage{bbm}
\usepackage{color}
\usepackage{lettrine}
\usepackage{xfrac}

\newtheorem{theorem}{Theorem}
\newtheorem{lemma}{Lemma}

\newtheorem{definition}{Definition}
\newtheorem{example}{Example}

\ifCLASSINFOpdf
\else
\fi
\begin{document}

\title{A Covert Queueing Channel in Round Robin Schedulers}

\author{AmirEmad~Ghassami,~\IEEEmembership{Student~Member,~IEEE,}
        Ali~Yekkehkhany,        and~Negar~Kiyavash,~\IEEEmembership{Senior~Member,~IEEE}
\thanks{The authors are with the Department of Electrical and Computer Engineering, and Coordinated Science Laboratory, University of Illinois at Urbana-Champaign, Urbana, IL 61801, USA (email: \texttt{\{ghassam2,yekkehk2,kiyavash\}@illinois.edu}).}}

\maketitle

\begin{abstract}
We study a covert queueing channel (CQC) between two users sharing a round robin scheduler. Such a covert channel can arise when users share a resource such as a computer processor or a router arbitrated by a round robin policy. We present an information-theoretic framework to model and derive the maximum reliable data transmission rate, i.e., the capacity of this channel, for both noiseless and noisy setups. Our results show that seemingly isolated users can communicate with a high rate over the covert channel and demonstrate the possibility of significant information leakage and privacy threats brought by CQCs in round robin schedulers. Moreover, we propose practical finite-length code constructions, which achieve the capacity limit.
\end{abstract}

\begin{IEEEkeywords}
Covert Queueing Channel, Round Robin Scheduler, Capacity Limit.
\end{IEEEkeywords}

\IEEEpeerreviewmaketitle

\section{Introduction}

\IEEEPARstart{S}{hared} resources among users in a system can lead to the occurrence of communication channels, which were not intended to exist in the system in the original design. 
Multiple users running on a computer who are using hardware resources such as CPU, storage, and multiple network streams flowing through a common router are examples of environments in which these channels can be created.
Such channels are mainly used by a malicious user for gaining access to other users' private information, and referred to as side channels
\cite{gong2011information, gong2014quantifying,kadloor2010low}.
For instance, the attacker can have an estimation of the features of the other users by analyzing their traffic.
Previous work shows that through traffic analysis, the attacker can obtain various private information including 
exact schedules of real-time systems \cite{chen2015schedule, chen2017reconnaissance},
visited web sites \cite{liberatore2006inferring}, sent keystrokes \cite{song2001timing}, and even inferring spoken phrases in a voice-over-IP connection \cite{wright2010uncovering}.

The coupling through shared resources can also be exploited for furtive communication among users. This type of channel is called a covert channel in the literature. Covert channels have typically been used by trusted insiders or malwares with access to secret information to leak it to untrusted outsiders \cite{tahir2015sneak, murdoch2005embedding, llamas2005evaluation, kang1996network}. 
In most of the work in the area of covert channels, two users communicate by modulating the timings, and the receiver sees a noisy version of the transmitter's inputs \cite{anantharam1996bits, ghassami2018covert, soltani2015covert, soltani2016covert, mukherjee2016covert}. Also, there are many works devoted to the detection of such channels \cite{cabuk2004ip,wsj2011,gianvecchio2007detecting}.
Unlike side channels, users in covert channels collaborate with each other and can agree on a specific usage pattern to efficiently utilize the features of the shared resource.

The focus of this paper is on the covert queueing channel (CQC), a special type of covert channel, that appears as a result of sharing a job scheduler among users. In a CQC, information is transmitted between the users through the delays that they experience when sending jobs to the shared scheduler.
More specifically, due to the inter-dependencies between delays observed by users, if one user experiences delays in service, the user becomes aware that the other users are issuing jobs \cite{ghassami2018covert}. Different scheduling policies, such as first-come-first-served (FCFS), time-division-multiple-access (TDMA), round robin, etc., can be used for resource allocation among users. The optimal scheme for message transmission and the rate of communication between the users depends on the used policy.

For scheduling policies, there is a trade-off between throughput and security.
For throughput, as long as the rates at which users request the shared resource are within the system's capacity region, an effective scheduler should be able to respond to the users' requests in a stable fashion. Such a scheduler is called a throughput optimal scheduler.
From a throughput optimality point of view, TDMA scheduling policy decouples the serving times of the users and hence, causes significant delays in the service given to the users. Therefore, this policy is not throughput optimal. However, from the security point of view, since users' delays are independent of each other in TDMA, this policy is the most secure scheduling policy \cite{gong2014quantifying}. 
The CQC created among users when the scheduler is FCFS is studied in detail in \cite{ghassami2015capacity,ghassami2018covert}. Although this scheduler does not waste any resource and hence is throughput optimal, it allows users to communicate with an information rate as high as 0.8114 bits per time slot.

In this paper we focus on a round robin scheduler, which is another throughput optimal policy commonly used in computer processors and communication networks.
Kadloor et al. \cite{kadloor2015delay} showed that when dealing with queueing side channels, round robin scheduling policy is privacy optimal within the class of work-conserving policies.
We focus on covert channels created in this scheduler. 
We present an information-theoretic framework to describe and model the data transmission in this channel and calculate its capacity.

Our system model is depicted in Figure \ref{SystemModelfig}. In this model we have an encoder and a decoder user, represented by Alice and Bob, respectively. There is no direct communication channel between the users, but they share a round robin scheduler. Hence, the delays observed by users are correlated. Therefore, Alice can encode a message in her traffic pattern and Bob can estimate the message by estimating Alice's traffic pattern via the delays he experiences.
We show that users can communicate with an information rate of 0.6942 bits per time slot through the covert channel created between them in this system in the absence of noise. 
Additionally, we study the noisy version of this covert channel in which packets are dropped with a certain probability, and we compute the capacity as a function of packet drop probability.

Followings are the main contributions of this work.
\begin{itemize}
\item We obtain the optimum signaling scheme for the CQC with round robin scheduler, and show that the capacity of the CQC is approximately 0.6942 bits per time slot (Section \ref{sec:noiseless}).
\item We propose an optimal finite block length coding scheme both when codewords are of fixed and of variable lengths. Our results show that the rates of the proposed optimal coding schemes approach the capacity as the number of messages goes to infinity (Section \ref{Finite}).
\item We extend the model to a more realistic noisy case, and calculate the capacity for this case as well (Section \ref{Noisy}).
\end{itemize}

\section{System Model}
\label{SystemModel}

We consider the system model depicted in Figure \ref{SystemModelfig}. In this model,  a shared resource services jobs from two users, Alice and Bob, using round robin policy.
 In this depiction, each packet is marked by its arrival time. As shown in this figure, there is a feedback line from the shared resource to the users, which notifies them when their packet is served. This allows the users to infer the status of the head of their queue; that is, at each time, the users will be aware that which one of their packets will be served next.

 Time is assumed to be discretized into slots, and the scheduler can serve one packet in each time slot. We follow the common convention that the packets arrive at the beginning of time slots and the departures occur at the end of time slots \cite{srikant2013communication,xie2016scheduling}. Each user's packets are buffered in a separate queue, and the round robin scheduler picks packets from the two queues as follows. In each time slot, three cases may happen:
(a) If both users' arrival queues are empty, the system remains idle and resumes scheduling in the next time slot.
(b) If only one user's queue has a packet, the current slot is given to that user, and the scheduler continues scheduling in the next time slot.
(c) If both users have waiting packets, the scheduler always gives priority to a fixed user. That is, the current time slot is allocated to serve a packet from the user with higher priority, and the next time slot will be allocated to the other user. The system continues scheduling after both users have received service. Without loss of generality, we assume that the priority is always given to Bob in the sequel.

We assume both Alice and Bob send at most one packet per time slot. Thus, their packet stream can be modeled as a binary bit stream, where bit `1' indicates a packet was sent, and bit `0' indicates no packet was transmitted. Since the scheduler can serve at most one packet per time slot, the sum of users' packet rates should be less than one for stability,
that is, $\lambda_1 + \lambda_2 < 1$, where $\lambda_1$ and $\lambda_2$ denote Alice and Bob's packet rates, respectively. (See Appendix \ref{AppendixA} for the proof of stability.)

\begin{figure}[t]
\centering
\includegraphics[scale=0.5]{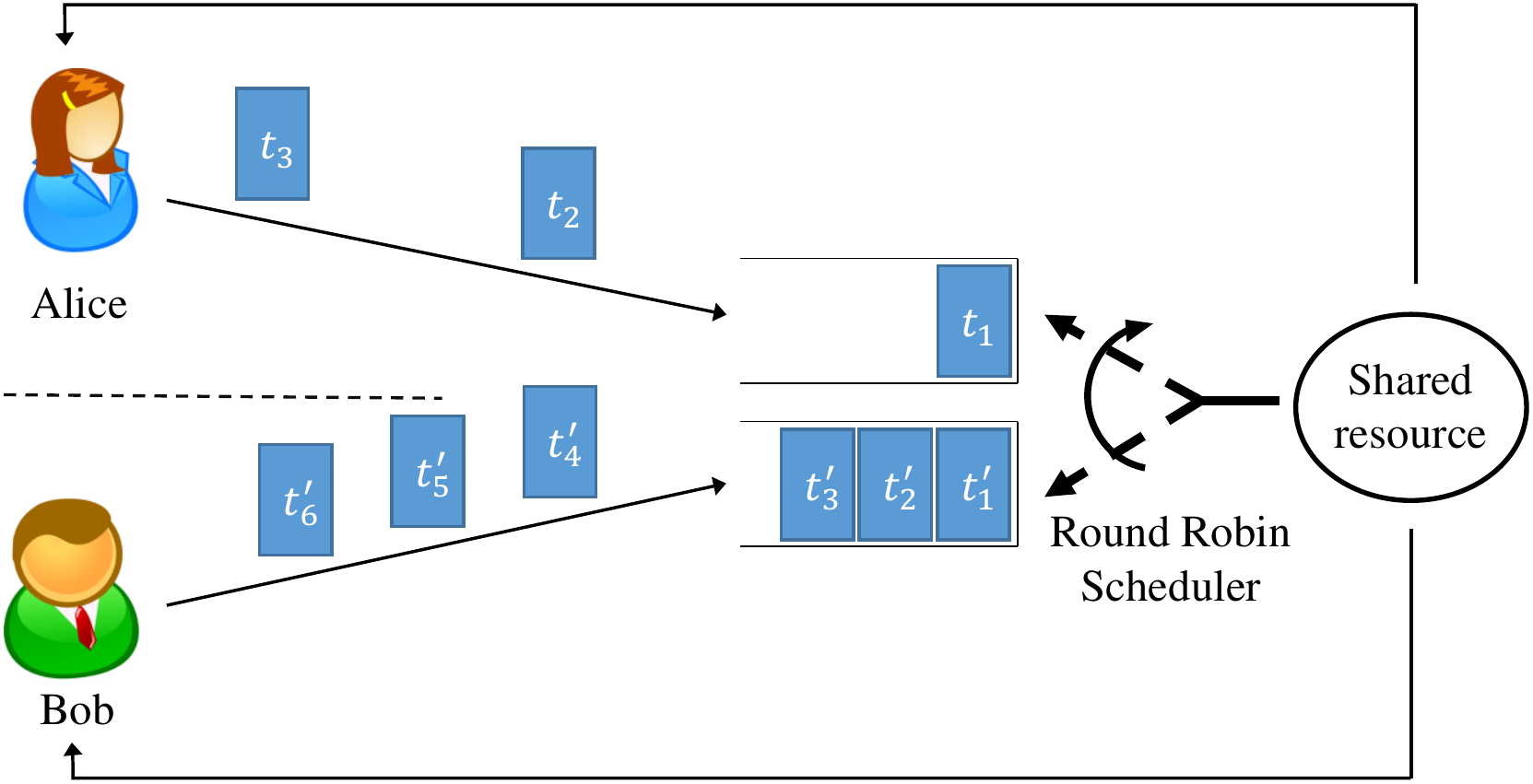}
\caption{System Setup: Alice and Bob share a resource arbitrated in round robin fashion. Users get acknowledgments when their packets are served.}
\label{SystemModelfig}
\end{figure}

Figure \ref{introexample} depicts an example of the scheduling in the system. In this and other such figures, Alice's and Bob's packets are shown by circled tip and regular arrows, respectively. For each user, the arrival stream, the head-of-the-queue stream and the departure stream are shown. Here, the arrival stream is the actual packet stream sent by the user, and the head-of-the-queue stream is the packets ready to be served at the head of the corresponding user's queue. Therefore, at any given time slot, the head of the queue can be `1' even though no packet has arrived in that slot. 
The downward streams in Figure \ref{introexample}(c) indicate the departure time of users' packets.
A packet is issued by both Alice and Bob in time slot 3. Since the priority is with Bob, his packet is processed in the same time slot, and Alice's packet is delayed by one time slot. Hence, Alice's head-of-the-queue stream contains a packet in both time slots 3 and 4. In Figure \ref{introexample}(b), the packet denoted by the gray dashed line indicates that it has been the same as its previous packet, which has been made to wait in the queue for one time slot to receive service in the next time slot. Alice's packet is processed in time slot 4, causing a delay for Bob's new packet sent in this time slot.

Suppose Alice aims to send message $W$ drawn uniformly from the set $\{1,2,\dots,M\}$. To this end, Alice encodes this message to a bit stream $X^m$ of length $m$, which is the codeword corresponding to message $W$. This codeword is sent out as a binary arrival stream $A_A^n$ of length $n$. In the same $n$ time slots, based on scheduling policy and both Alice's and Bob's packet arrivals, Bob receives a binary acknowledgment stream from the system, which is denoted by $D_B^n$. Finally, Bob transforms this stream to a bit stream $Y^m$ which will be decoded to message $\hat{W}$. As a result, we have the following Markov chain:
\begin{equation}
\label{MarkovChain}
W \rightarrow X^m \rightarrow A_A^n \rightarrow D_B^n \rightarrow Y^m \rightarrow \hat{W}
\end{equation}
The noise in the system is modeled as follows. The packets generated by either Alice or Bob may be dropped in the link between the users and the shared resource with probability $\delta$. Note that this noise can affect the transmissions in $X^m \rightarrow A_A^n$ and $A_A^n \rightarrow D_B^n$ in Markov chain \eqref{MarkovChain}. In Section \ref{sec:noiseless} we will obtain the optimum signaling scheme between Alice and Bob, and will show that using this scheme, noise will not affect the transmission $A_A^n \rightarrow D_B^n$.

\section{Coding Theorem}
\label{sec:noiseless}

In this section we obtain the optimum signaling scheme between the users and calculate the capacity of the introduced covert channel. 
The performance metric used in defining the capacity is the average error probability, defined as follows.
\[
P_e \triangleq P(W \neq \hat{W}) = \sum_{m = 1}^{M} \frac{1}{M} {P} (\hat{W} \neq m | W = m).
\]
Next, we define concepts required for the coding theory.
These definitions are natural extensions of the classical definitions in information theory \cite{anantharam1996bits,cover2012elements,csiszar2011information}.
\begin{definition}
An $(n, M, \epsilon)$-code consists of a codebook of size $M$ with equiprobable binary codewords of average length $n$ satisfying $P_e \leq \epsilon$.
\end{definition}

\begin{definition}
The information transmission rate of a code is
$R = \log M/n,$
which is the amount of conveyed information normalized by the average number of used time slots\footnote{Throughout the paper, all the logarithms are in base 2.}.
\label{Def2}
\end{definition}
\begin{definition}
A rate $R$ is said to be achievable if there exists a sequence of $(n, M, \epsilon_{n})$-codes such that $\epsilon_{n} \rightarrow 0$ as $n \rightarrow \infty$.
\end{definition}
\begin{definition}
The channel capacity is the supremum achievable rate at which Alice can communicate through the covert channel with Bob.
\end{definition}
We first obtain the optimum signaling scheme between the users, which  maximizes the information transmission.
\begin{lemma}
\label{lem:hotq}
In the CQC in Figure \ref{SystemModelfig}, the maximum information transmission rate between the users is achieved when Bob's head-of-the-queue bit stream is always equal to `1'.
\end{lemma}

See Appendix \ref{app:hotq} for a proof.\\

The requirement that Bob should have a packet in his head-of-the-queue at all time slots does not mean that he has to send a packet in all time slots. It suffices for him to fix his queue length at some nonzero length, and whenever one of his packets is served, he generates a packet to ensure his queue length remains nonzero. This strategy allows him to keep the sum rate of arrivals from Alice and Bob less than 1 and keep the system stable.

As stated in the proof of Lemma \ref{lem:hotq}, in an optimal signaling, Alice can signal two distinguishable patterns to Bob. These patterns require Alice to have either a `0' or a `1' in her head-of-the-queue stream. We will next demonstrate the arrival stream required for these head-of-the-queue streams.

For Alice to have a `0' in her head-of-the-queue stream, assuming the queue is empty, she needs to idle for one time slot. Therefore, $A_A$ should contain a `0'. If Alice has a `1' in her head-of-the-queue stream, two time slots are required for this signal to be transmitted (one for Bob's packet and one for Alice's). Therefore, $A_A$ should contain a `10'.
Note that Alice should not send two packets in two consecutive time slots. This is because if Alice sends two (or more) packets in consecutive time slots, her next (or more) `0'(s) would disappear as her packets are accumulated in the queue. This is demonstrated in Example \ref{example111}.

\begin{figure}[t]
\centering
\includegraphics[scale=0.85]{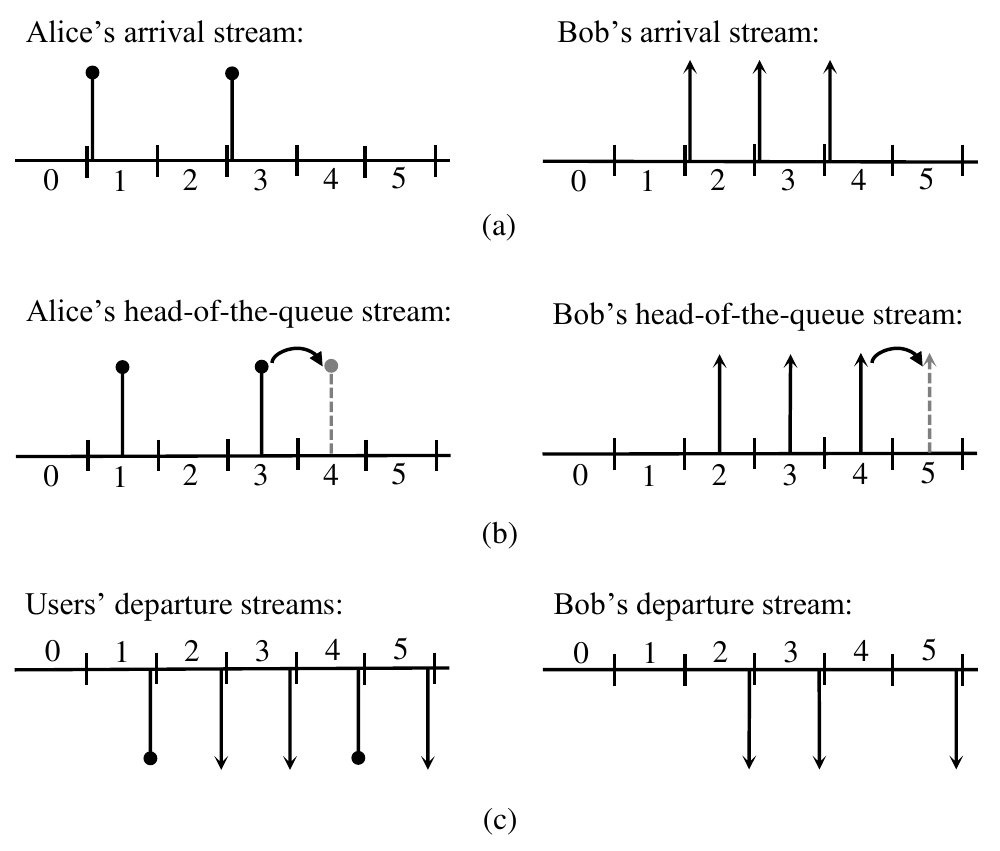}
\caption{(a) Arrival streams; (b) Head-of-the-queue streams; (c) Departure streams.}
\label{introexample}
\end{figure}

\begin{example}
\label{example111}
Assume that Alice's queue is empty and she wants to transmit the bit stream $1101$ to Bob. If Alice sends `1's in her message consecutively in each time slot (as depicted in Figure \ref{fig:BB}(a)), Bob would erroneously decode the message `111'. This is caused by the accumulation of packets in Alice's queue, stemming from existence of a packet at the head of her queue before clearing the previous `1'. Figure \ref{fig:BB}(b) depicts the correct signaling by Alice.
\end{example}
Therefore, the optimum signaling scheme from Alice to Bob which maximizes the information transmission rate would be summarized as follows:
\begin{itemize}
\item Signaling bit `1':
To signal bit `1' in time slot $n$, Alice must have a head-of-the-queue packet at the beginning of the time slot.
Recall that Bob has a ready-to-be-served packet in all time slots. Thus, round robin policy will serve Bob and Alice at time slots $n$ and $n+1$, respectively. Therefore, when Bob receives service in a time slot but does not receive service in the next time slot, he decodes bit `1'.
\item Signaling bit `0':
To signal bit `0' in time slot $n$, Alice must not have a head-of-the-queue packet at the beginning of the time slot.
Because Bob has a packet which is ready to be serviced in the head of the queue in this time slot, he receives service at time slot $n$, and the scheduler resets for time slot $n+1$.
As a result, at time slot $n + 1$, Bob is served again. Therefore, if Bob receives service in two consecutive time slots, he decodes it as bit `0'.
\end{itemize}
This scheme implies the following lemma.
\begin{lemma}
\label{lem:frac}
In the CQC in Figure \ref{SystemModelfig}, in the scheme with the maximum information transmission rate between the users, we have
\[
\frac{m}{n}=\frac{1}{1+p},
\]
where $p=P(X=1)$.
\end{lemma}
\begin{proof}
According to the optimum signaling scheme, sending a bit `0' and a bit `1' require 1 and 2 time slots respectively. Therefore,
\[
\frac{m}{n}=\frac{m}{2pm+(1-p)m}=\frac{1}{1+p}.
\]
\end{proof}
Equipped with Lemma \ref{lem:frac}, we next calculate the capacity of the introduced covert channel.
\begin{theorem}
\label{theorem111}
The capacity of the introduced CQC in a shared round robin scheduler in Figure \ref{SystemModelfig} is
\begin{equation}
\label{cap1}
C =\max_{0\le p\le1} \frac{h(p)}{1 + p},
\end{equation}
where $p$ is the probability of sending message bit `1' by Alice and $h(\cdot)$ is the binary entropy function.
\end{theorem}
The proof of Theorem \ref{theorem111} is presented in Appendix \ref{app:noiseless}. The maximum of \eqref{cap1} is approximately $0.6942$ achieved at $p = \frac{3 - \sqrt{5}}{2}$.

\begin{figure}[t]
\centering
\includegraphics[scale=0.8]{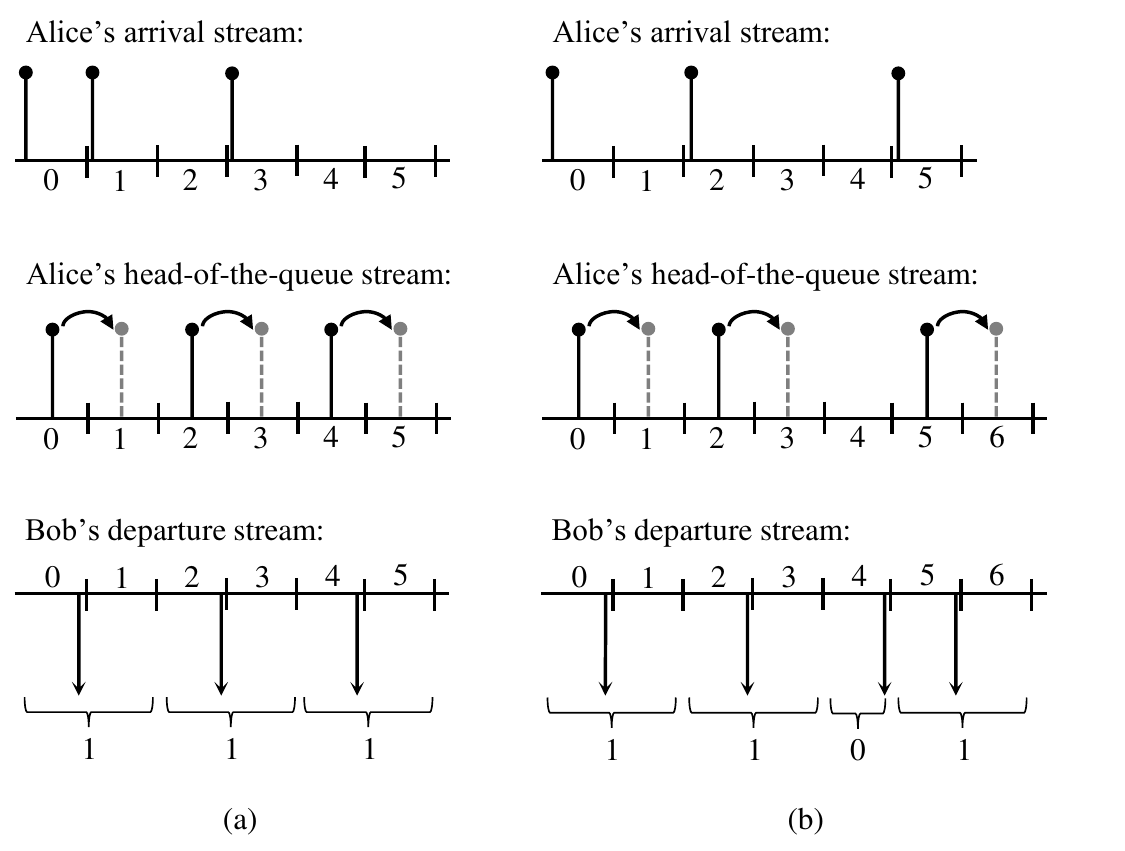}
\caption{
Visualization of Example \ref{example111}.}
\label{fig:BB}
\end{figure}

\section{Finite-length Codeword Regime}
\label{Finite}

We proposed an achievable scheme in the proof of Theorem \ref{theorem111}; however, this scheme requires the value of $n$ to tend to infinity to achieve the capacity.
In this section we obtain the optimum coding schemes in finite-length codeword regime for communication in the introduced covert queueing channel. Our proposed schemes achieve information rates close to the capacity even with small codebooks.

As mentioned earlier, Alice encodes each message to a binary sequence and creates a codebook $\mathcal{C}$, known to both Alice and Bob. The codewords in the codebook could be all of the same or different lengths.
In the following two subsections, we will consider both these scenarios and find the optimum codebook for the setting. 

\subsection{Variable-length Codewords}

In this subsection, for any fixed number of messages $M$, we propose an algorithm which generates the optimum variable-length codebook, i.e., the list of codewords that results in maximum communication rate between the users. By Definition \ref{Def2}, the information rate at which Alice communicates with Bob could be computed as 
\begin{equation}
\label{rate1}
\begin{aligned}
R = \frac{\log M}{\frac{1}{M}\sum_{m = 1}^{m = M}T_m},
\end{aligned}
\end{equation}
where $T_m$ is the transmission time of the $m$-th codeword. As we discussed in Section \ref{SystemModel}, in an optimum signaling scheme, transmission of bit `1' takes two time slots, while bit `0' requires one time slot. Denote the number of bits `0' and `1' in codebook $\mathcal{C}$ by $n_0(\mathcal{C})$ and $n_1(\mathcal{C})$, respectively. Therefore, $\sum_{m = 1}^{m = M}T_m = 2 n_1(\mathcal{C}) + n_0(\mathcal{C})$, and \eqref{rate1} could be rewritten as
\begin{equation}
\begin{aligned}
R = \frac{M \log M}{2 n_1(\mathcal{C}) + n_0(\mathcal{C})}.
\label{rateM}
\end{aligned}
\end{equation}

Given that $M$ is a fixed given parameter, maximizing the rate is equivalent to searching for a codebook which achieves the minimum of the denominator in \eqref{rateM}.

\begin{figure}[t]
\centering
\includegraphics[scale=0.55]{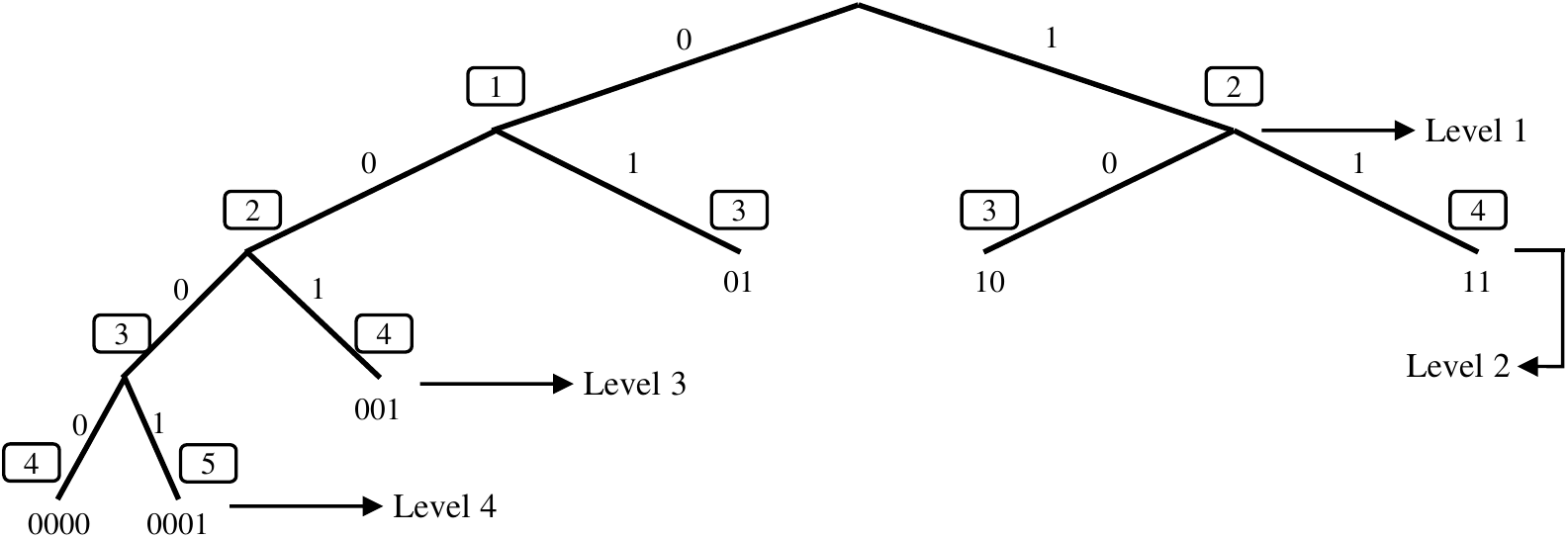}
\caption{A tree representation of the codewords. Codewords are the leaves of the tree. The cost of a codeword, defined to be its transmission time, is written in the boxes.}
\label{TreeCodeWord}
\end{figure}

Our technique for finding the optimum codebook is as follows. We represent each codeword in the codebook by a leaf in a tree, as depicted in Figure \ref{TreeCodeWord}. 
We define the cost of the codeword $\bar{X}$ as the number of time slots required for transmission of this codeword, and denote it by $\eta(\bar{X})$. That is, $\eta(\bar{X})=2 n_1(\bar{X}) + n_0(\bar{X})$, where $n_0(\bar{X})$ and $n_1(\bar{X})$ denote the number of bits `0' and `1' in $\bar{X}$, respectively.
The numbers in boxes in Figure \ref{TreeCodeWord} denote the cost of each codeword.
We call the resulting graph, the \textit{codeword tree}. In this representation, for each node, the branch to the left (right) side, appends a 0 (1) to the codeword corresponding to that node. For example, if a vertex represents the bit stream $00101$, its left and right children will represent codewords $001010$ and $001011$, respectively. The reason we use the leaves of a tree for representing the codewords is to guarantee that the codewords are uniquely decodable \cite{cover2012elements}. 

\begin{figure}[t]
\centering
\includegraphics[scale=0.44]{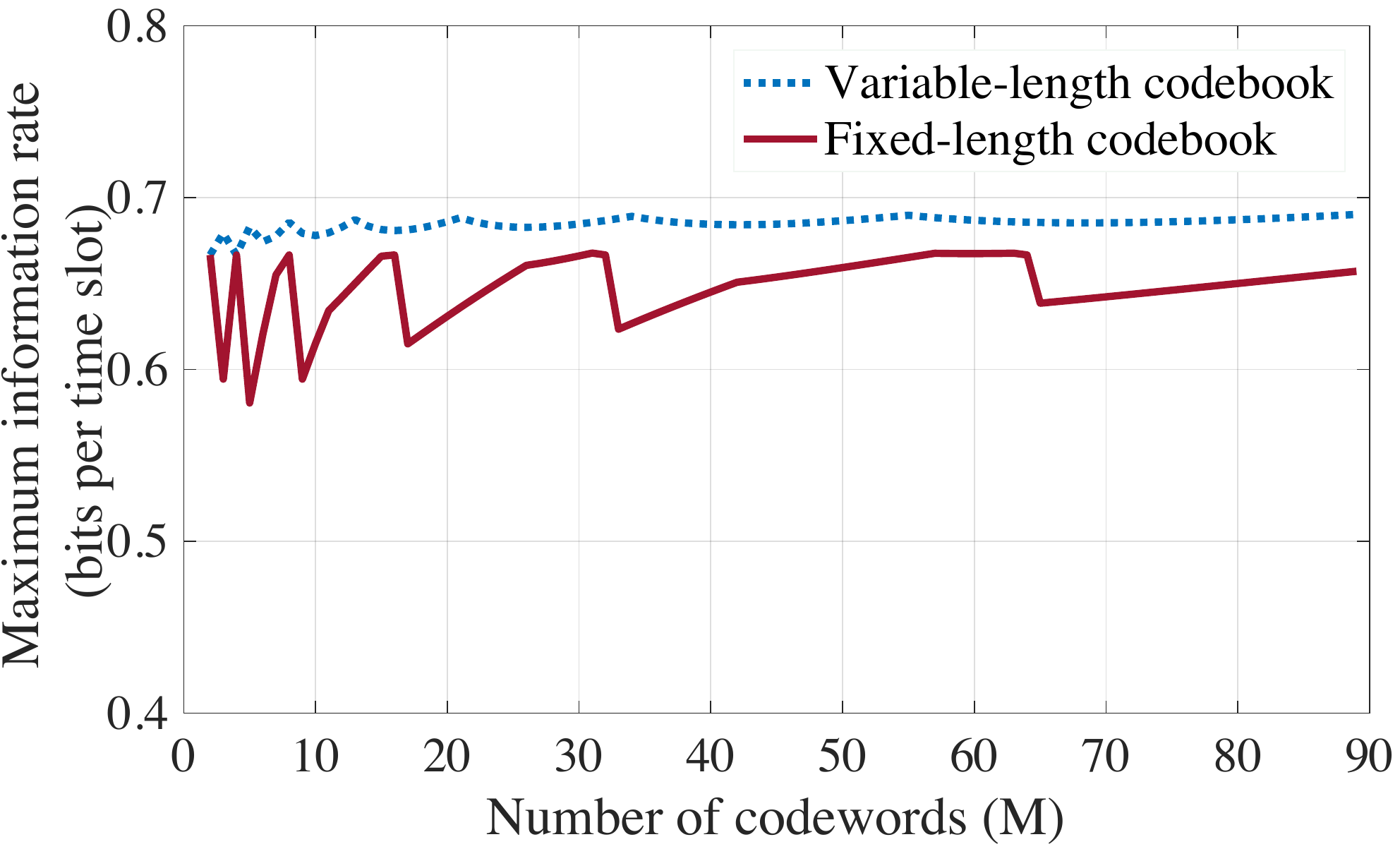}
\caption{The maximum rates at which Alice can communicate with Bob versus the number of codewords for two cases of variable- and fixed-length codebooks.}
\label{MaxRateRR2}
\end{figure}

Algorithm \ref{algo1} describes how the $M$ optimal codewords are selected from the tree. In this algorithm, we initialize the codebook to be $\{0,1\}$. In each iteration, one of the current codewords $\bar{X}$ with minimum cost is replaced with its two children $\bar{X}0$ and $\bar{X}1$. This procedure is repeated until all $M$ codewords are obtained. As an example,
the result of Algorithm \ref{algo1} for $M=6$ is depicted in Figure \ref{TreeCodeWord}.

\begin{algorithm}[t]
\caption{Finding the optimum variable-length codebook with $M$ codewords}
\label{algo1}
\begin{algorithmic}[1]
\STATE \textbf{Initialize} $\mathcal{C}=\{0, 1 \}$
\WHILE {$|\mathcal{C}|\le M$}
	\STATE Find $\bar{X}^*=\arg\min_{\bar{X}\in\mathcal{C}}\eta(\bar{X})$
	\STATE $\mathcal{C}=\mathcal{C}\setminus\{\bar{X}\}$
	\STATE $\mathcal{C}=\mathcal{C}\cup\{\bar{X}0\}\cup\{\bar{X}1\}$
\ENDWHILE
\STATE \textbf{Output} $\mathcal{C}$
\end{algorithmic}
\end{algorithm}

\begin{theorem}
\label{optcode1}
For a fixed given number of equiprobable messages, Algorithm \ref{algo1} is optimal in the sense that it provides a codebook which maximizes the communication rate between the users.
\end{theorem}
A proof of Theorem \ref{optcode1} is presented in Appendix \ref{app:optcode1}.\\

The maximum rate at which Alice can communicate with Bob versus the number of codewords $M$, is depicted in Figure \ref{MaxRateRR2}. The overall trend of the maximum communication rate increases as the number of codewords increases, and converges to the capacity. The following theorem formalizes this claim.

\begin{theorem}
\label{rateconv1}
The information transmission rate of a codebook created by Algorithm \ref{algo1} converges to the capacity of the covert channel as the number of messages goes to infinity.
\end{theorem}
A proof of Theorem \ref{rateconv1} is presented in Appendix \ref{app:rateconv1}.

\subsection{Fixed-length Codeword}
\label{FixedLengthCodeword}

In many applications, using variable-length codewords is not desirable from the designer's point of view. For example, in a noisy system, a variable-length scheme may lead to loss of synchronization between encoder and decoder. To obtain fixed-length codewords, all $M$ codewords must be selected from the same level of the codeword tree. Such a constraint on choosing codewords can lead to reduction in information rate for a fixed number of messages, however, as we shall see, these codes can still achieve the capacity when the length of the codewords goes to infinity.

\begin{algorithm}[t]
\caption{Finding the optimum fixed-length codebook with $M$ codewords}
\label{algo2}
\begin{algorithmic}[1]
\STATE \textbf{Initialize} $\hat{l} = \lceil log(M) \rceil$
\FOR {\emph{$l = \hat{l}$ to $2\hat{l}$}}
	\STATE $\mathcal{C}_l$ = Set of $M$ codewords with the least number of `0's in the $l$-th level of the codeword tree
	\STATE $\eta(\mathcal{C}_l) = n_0(\mathcal{C}_l) + 2 n_1(\mathcal{C}_l)$
\ENDFOR
\STATE $l^* = \arg \min_l \eta(\mathcal{C}_l)$
\STATE \textbf{Output} $\mathcal{C}_{l^*}$
\end{algorithmic}
\end{algorithm}

Our proposed approach for selecting the optimal fixed-length codebook for a given number of messages is presented in Algorithm \ref{algo2}.
Denote the cost of a codebook with $\eta(\mathcal{C})= n_0(\mathcal{C}) + 2 n_1(\mathcal{C})$, which is the sum of the cost of its codewords.
In Algorithm \ref{algo2} first the optimum codebooks with the least number of `1's should be chosen in each of the levels $\hat{l} = \lceil \log (M) \rceil$ to $2\hat{l}$. The optimal codebook is then the one with minimum cost among these created codebooks.

\begin{theorem}
\label{optcode2}
For a fixed given number of equiprobable messages, Algorithm \ref{algo2} outputs the optimal fixed-length codebook.
\end{theorem}
A proof of Theorem \ref{optcode2} is presented in Appendix \ref{app:optcode2}.\\

As shown in Figure \ref{MaxRateRR2}, using Algorithm \ref{algo2} the overall trend of the maximum communication rate increases as the number of codewords increases, and converges to the capacity. The following theorem formalizes this claim.

\begin{theorem}
\label{rateconv2}
The information transmission rate of a codebook created by Algorithm \ref{algo2} converges to the capacity of the covert channel as the number of messages goes to infinity.
\end{theorem}
A proof of Theorem \ref{rateconv2}, is presented in Appendix \ref{app:rateconv1}.\\

The maximum rate at which Alice can communicate with Bob versus the number of codewords $M$, in the case of using fixed-length codewords is depicted in Figure \ref{MaxRateRR2}. 


\section{Noisy Channel Case}
\label{Noisy}

In this section we consider the case where the channel between the users is noisy.
The noise model is as follows. 
The packets generated by either Alice or Bob may be dropped in the link between the users and the shared resource with probability $\delta$.
We start by investigating the effect of the noise on the CQC between the users.
The following definition is required in our analysis.
\begin{definition}
A Z-channel with parameter $\delta$ is a discrete memoryless channel, in which `0' is always transmitted error-free, but `1' is flipped with probability $\delta$ (Figure \ref{TZchannel}).
\end{definition}
For the case of noisy channel, Lemma \ref{lem:hotq} again holds and hence, the signaling scheme proposed in Section \ref{sec:noiseless}, achieves the maximum information transmission rate in the noisy channel as well. As discussed in Section \ref{sec:noiseless}, Bob should keep his queue length positive at all time slots. 
If he keeps his queue length large enough, even if his packets are dropped in multiple time slots, he still has remaining ready to be served packets in his queue.
Thus by keeping the probability of his queue length becoming zero arbitrary small, Bob can avoid dropped packets impacting the scheme. 
Therefore, although noise can affect the transmissions in $X^m \rightarrow A_A^n$ in Markov chain \eqref{MarkovChain}, it will not influence the transmission $A_A^n \rightarrow D_B^n$.

\begin{figure}[t]
\centering
\includegraphics[scale=0.93]{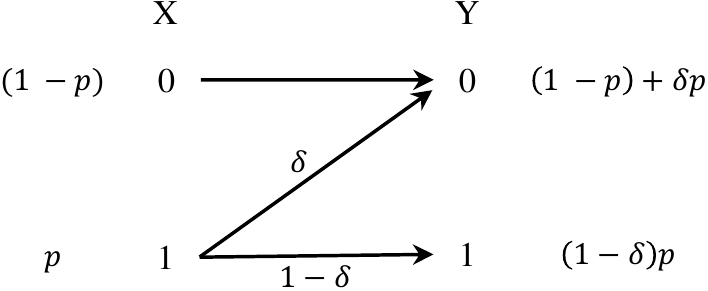}
\caption{Z-channel model of the covet channel with packet drops.}
\label{TZchannel}
\end{figure}

\begin{lemma}
\label{lemma:drop}
In the proposed optimum signaling scheme proposed in Section \ref{sec:noiseless}, packet drops convert the channel between Alice and Bob to a Z-channel with parameter $\delta$.
\end{lemma}
\begin{proof}
Noise does not affect the data transmission when Alice sends `0' as she does not send any packets in this case. On the other hand, when Alice sends a packet to communicate bit `1', this bit will change to `0' if the packet is dropped, which happens with probability $\delta$.
As stated in Section \ref{sec:noiseless}, normally Alice should wait for one time slot after she sends a packet; however, when packet drops occur she does not need to wait for a time slot. Alice can always tell that a drop has occurred because she knows her queue length at the end of each time slot.
Thus the aggregate effect of noise may be modeled as a Z-channel.
\end{proof}
Figure \ref{TZchannel} shows the resulting Z-channel. Note that the channel is depicted between $X$ and $Y$ in Markov chain \eqref{MarkovChain}, but the noise occurs between sequences $X$ and $A_A$.

In the noisy channel case, Lemma \ref{lem:frac} will be modified as follows.
\begin{lemma}
\label{lem:frac2}
In the noisy CQC in Figure \ref{SystemModelfig}, in the scheme with the maximum information transmission rate between the users, we have
\[
\frac{m}{n}=\frac{1}{(1-p)+\delta p+2(1-\delta)p},
\]
where $p=P(X=1)$.
\end{lemma}
\begin{proof}
According to the optimum signaling scheme, sending a bit `0' and a bit `1' require 1 and 2 time slots, respectively. A bit `0' can be a result of either a send `0' or a flipped `1'. Therefore,
\begin{align*}
\frac{m}{n}&=\frac{m}{\delta pm+(1-p)m+(1-\delta)2pm}\\
&=\frac{1}{(1-p)+\delta p+2(1-\delta)p}.
\end{align*}
\end{proof}

Equipped with Lemmas \ref{lemma:drop} and  \ref{lem:frac2}, we next calculate the capacity of the introduced noisy covert channel.

\begin{figure}[t]
\centering
\includegraphics[scale=0.45]{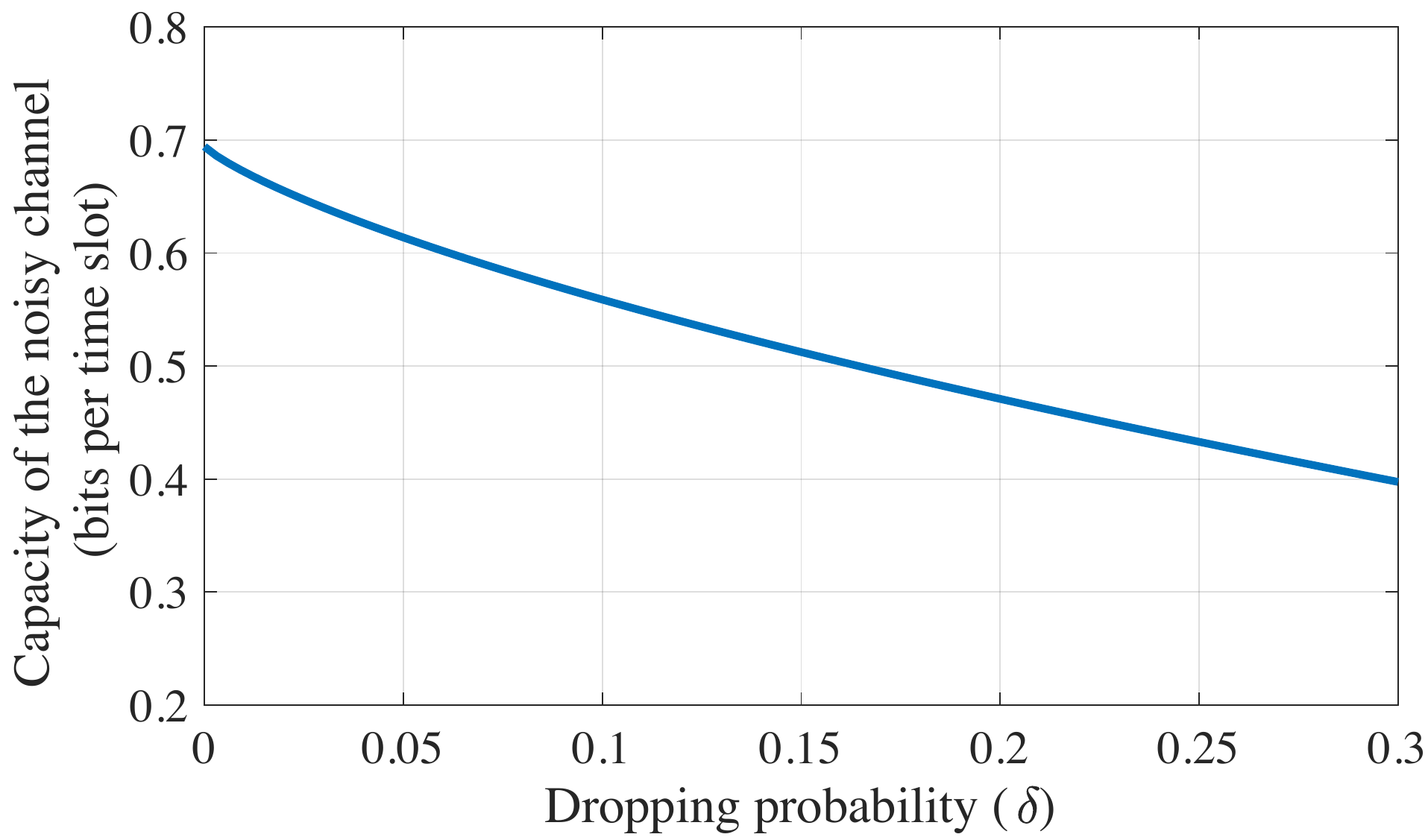}
\caption{Capacity of the noisy CQC versus the drop probability $\delta$.}
\label{MaxRateNoise}
\end{figure}

\begin{theorem}
\label{theorem2}
The capacity of the noisy CQC between Alice and Bob with packet drop probability $\delta$ in a shared round robin scheduler in Figure \ref{SystemModelfig} is
\begin{equation}
C = \underset{0\le p\le1}{\mathrm{max}} \frac{h((1 - \delta) p) - p h(\delta)}{(1 - p) + \delta p + 2 (1 - \delta) p},
\end{equation}
where $p$ is the probability of sending message bit `1' by Alice and $h(\cdot)$ is the binary entropy function.
\end{theorem}
See Appendix \ref{AppendixB} for a proof.\\

Figure \ref{MaxRateNoise} depicts the capacity of the noisy CQC versus the dropping probability $\delta$.

In the noisy setting, as mentioned earlier, synchronization between the encoder and decoder sides of the system may be lost. To prevent this from happening, users can use a fixed-length codebook such as the one presented in Subsection \ref{FixedLengthCodeword}.

\section{Conclusion}

We studied a covert queueing channel between two users sharing a round robin scheduler.
An information-theoretic framework was proposed to derive the capacity of this channel in both noisy and noiseless cases.
We showed that in the noiseless case, an information rate as high as 0.6942 bits per time slot is achievable in this channel. Clearly this rate of transmission can lead to significant information leakage in a system and deserves special attention in high security systems.
For the noisy case, where users' packets may drop, we again analyzed the highest achievable information rate and obtained the capacity for different levels of noise.
Moreover, we proposed practical finite-length code constructions,
which asymptotically achieves the capacity limit.


\begin{appendices}

\section{Proof of Stability}
\label{AppendixA}


As mentioned in Section \ref{SystemModel}, each user has a separate queue. Denote the queue length and the number of packet arrivals at each queue at time slot $n$ by $q_i(n)$ and $a_i(n)$, respectively. Let $\mathbb{E} [a_i(n)] = \lambda_i$ and $\mathbb{E} [a_i^2(n)] < \infty$, where $i\in\{A,B\}$ signifies Alice or Bob.
We assume the arrival processes of Alice and Bob are independent of each other and the system state. The system is stable if neither user's queue length grows to infinity in the steady state of the system, as long as the arrivals are in the capacity region of the scheduler. Thus, it suffices to prove that the sum of the queue lengths is finite with probability one, which implies the stability of both queues. We use the Foster-Lyapunov theorem to prove this statement \cite{yekkehkhany2018gb}. 
Denote the sum of the queue lengths with $\hat{q} (n) = q_A(n) + q_B(n)$, and the sum of packet arrivals for Alice and Bob with $\hat{a} (n) = a_A(n) + a_B(n)$. Note that $\mathbb{E}[\hat{a}(n)] = \lambda_A + \lambda_B$ and the second moment of $\hat{a}(n)$ is finite. As long as a task is available in one of the two queues, the round robin scheduler serves a task; that is, the service rate is one packet per time slot. Thus $\hat{q}$ evolves as
\begin{equation*}
\hat{q} (n + 1) = (\hat{q}(n) + \hat{a}(n) - 1)_+,
\end{equation*}
where $(x)_+ \triangleq \max \{x, 0 \}$.

Choose the Lyapunov function $V(\hat{q}(n)) = \frac{\hat{q}^2(n)}{2}$. Note that this choice of Lyapunov function satisfies the requirements of non-negativity, being equal to zero only at $\hat{q} = 0$, and going to infinity as $\hat{q}$ goes to infinity \cite{yekkehkhany2017near}. We show that the drift of this Lyapunov function is negative outside of a bounded region of the state space, and is positive and finite inside this bounded region, which implies that the system state is positive recurrent.

\begin{equation*}
\begin{aligned}
\mathbb{E} [ & V(\hat{q} (n + 1)) - V(\hat{q} (n)) \big | \hat{q}(n) = q] \\
& = \mathbb{E} [\frac{(q + \hat{a}(n) - 1)_+^2 - q^2}{2} \big | q] \\
& \leq \mathbb{E} [\frac{ \hat{a}^2(n) + 1 + 2\hat{a}(n) q - 2q - 2\hat{a}(n)}{2}] \\
& \leq \mathbb{E} [\frac{ \hat{a}^2(n) + 1 - 2\hat{a}(n)}{2} - 2q + 2\hat{a}(n) q] \\
& = c - 2q(1 - \lambda_A - \lambda_B),
\end{aligned}
\end{equation*}
where in the last equality, $c$ is a constant because $\hat{a}$ has bounded first and second moments.
For $\lambda_A + \lambda_B < 1$, the drift of the Lyapunov function is bounded by the constant $c$ in the bounded set $B = \{ \hat{q} \ | \ \hat{q} \leq \frac{c}{2(1 - \lambda_A - \lambda_B)} \}$, and is negative in the complement set, $B^c$. Therefore, with our queueing structure and round robin scheduler, the system is stable as long as $\lambda_A + \lambda_B < 1$.

\section{Proof of Lemma \ref{lem:hotq}}
\label{app:hotq}
Suppose Bob's head-of-the-queue stream contains a `1' followed by $t$ `0's. During these $t+1$ bits, if Alice's head-of-the-queue is equal to $t+1$ consecutive `1's, $D_B$ will contain a `1' followed by $t+1$ `0's; otherwise, it will contain a `1' followed by $t$ `0's. Therefore, denoting the probability of the event that Alice's head-of-the-queue is equal to $t+1$ consecutive `1's by $p_{t+1}$, there are two distinguishable outputs for Bob, received on average in $(t+1)(1-p_{t+1})+(t+2)p_{t+1}=(t+1)+p_{t+1}$ time slots, which gives the information transmission rate of
\[
\frac{\log 2}{(t+1)+p_{t+1}},
\]
which regardless of the value for $p_{t+1}$, is maximized when $t=0$. Therefore, the maximum information transmission rate between the users is achieved when Bob's head-of-the-queue bit stream is always equal to `1'.

\section{Proof of Theorem \ref{theorem111}}
\label{app:noiseless}

The proof consists of achievability and converse arguments.\\
\textbf{Converse:} For the CQC in a shared round robin scheduler with service rate 1 depicted in Figure \ref{SystemModelfig}, any code consisting of a codebook of $M$ equiprobable binary codewords, where messages take on average $n$ time slots to be received, satisfies
\begin{equation*}
\begin{aligned}
\frac{1}{n}\log{M}&\stackrel{(a)}{=}\frac{1}{n}H(W)
=\frac{1}{n}H(X^m)
\le \frac{1}{n}\sum_{i=1}^mH(X_i)\\
&\le \max_{P_X}\frac{m}{n}H(X)
\le\max_{0\le p\le 1}\frac{m}{n}h(p),
\end{aligned}
\end{equation*}
where $(a)$ holds because $W$ is a uniform random variable over the message set $\{1,...,M\}$. By Lemma \ref{lem:frac}, in the scheme with the maximum information transmission rate between Alice and Bob, we have $\frac{m}{n}=\frac{1}{1+p}$. Therefore, 
\[
C\le\max_{0\le p\le 1}\frac{h(p)}{1+p}.
\]
\textbf{Achievability:}
We fix a Bernoulli distribution $P$ with parameter $p^*$, where
\begin{equation*}
\begin{aligned}
p^* = \arg \max_{p\in[0,1]} \frac{h(p)}{1+p},
\end{aligned}
\end{equation*}
and generate a binary codebook $\mathcal{C}$ containing i.i.d. sequences of length $m$ drawn according to $P$, where $m=\frac{n}{1+p^*}$. A standard typicality argument \cite[Chapter 3]{cover2012elements}, shows that as $n$ goes to infinity, we can have $2^{H(X^m)}=2^{mh(p^*)}$ distinct codewords in $\mathcal{C}$.

In order to send a bit `1', Alice sends a packet and then idles for one time slot. To send a bit `0', she just idles for one time slot. Thus, each message on average takes $2mp^*+m(1-p^*) = n$ time slots to be transmitted. At the same time, Bob keeps his head of the queue always full.
Since this is a noiseless channel, the error will be zero. Therefore,
\[
C\ge \frac{\log 2^{mh(p^*)}}{n}=\frac{mh(p^*)}{m(1+p^*)}=\frac{h(p^*)}{1+p^*}=\max_{0\le p\le 1}\frac{h(p)}{1+p}.
\]
The achievability and converse parts above, complete the proof of the coding theorem.

\section{Proof of Theorem \ref{optcode1}}
\label{app:optcode1}

We show that Algorithm \ref{algo1} minimizes the sum of costs of codewords which is $2 n_1(\mathcal{C}) + n_0(\mathcal{C})$. 

We note that, replacing codeword $\bar{X}$ with cost $\eta(\bar{X})$ results in two codewords $\bar{X}0$ and $\bar{X}1$ with costs $\eta(\bar{X}) + 1$ and $\eta(\bar{X}) + 2$, respectively. As a result, replacing codeword $\bar{X}$ with its two children
causes additional cost of $\eta(\bar{X}) + 3$, and an additional codeword to the codebook. Therefore, since the added cost is increasing in $\eta(\bar{X})$, to obtain the optimal codebook in a step, it suffices to replace the minimum cost codeword by its children.

Suppose Algorithm \ref{algo1} outputs codebook $\mathcal{C}_1$. To prove the claim by contradiction, suppose the codebook $\mathcal{C}_2$ resulted from another algorithm is optimum where both $\mathcal{C}_1$ and $\mathcal{C}_2$ have $M$ codewords.
We first find the subtree which is common between $\mathcal{C}_1$ and $\mathcal{C}_2$, which implies that two algorithms are equivalent until, say, step $m$.
From that step, all the replacements are different in two algorithms. The first replacement in Algorithm \ref{algo1} gives a smaller cost (because we assumed to choose the minimum cost replacement). For the next replacement in step $m+1$, Algorithm \ref{algo1} had the option of the other algorithm's replacement in step $m$, yet it did not choose that. This means that again a better replacement was possible.
Adding this to the fact that the costs of children of a codeword is larger than its own cost, concludes that the replacement in step $m+1$ for Algorithm \ref{algo1} was also a better choice. This reasoning applies to all steps in which two algorithms are different and concludes that $\mathcal{C}_2$ cannot be optimum.

\section{Proof of Theorems \ref{rateconv1} and \ref{rateconv2}}
\label{app:rateconv1}

In this appendix, we show that the information transmission rate of our proposed optimum codebooks created by Algorithms \ref{algo1} and \ref{algo2} converge to the capacity of the covert channel as the number of messages goes to infinity.
We prove that the information rate of another non-optimum codebook with rate lower than the rates of both aforementioned codebooks achieves the capacity.

Consider a codebook with fixed-length codewords from the $l$-th level of the codeword tree. We choose each codeword to have exactly $\lfloor lp\rfloor $ bits `1', where the parameter $p \in [0, 1]$ can be selected in a manner to maximize the information rate. Such a codebook consists of $M = {{l}\choose{\lfloor lp \rfloor}}$ messages all of which have equal transmission time $2 \times \lfloor lp \rfloor + (l - \lfloor lp \rfloor)$. We show that the information rate of this codebook asymptotically converges to the capacity as $l$ (or equivalently the number of messages) goes to infinity.
From Definition \ref{Def2}, we have
\[
\begin{aligned}
\max_p \lim_{n \rightarrow \infty} \frac{\log{{l}\choose{\lfloor lp \rfloor}}}{n} 
 &\overset{(a)}{=} \max_p \lim_{l \rightarrow \infty} \frac{l\cdot h(\frac{\lfloor lp \rfloor}{l}) + o(l)}{ l + \lfloor lp \rfloor } \\
&\overset{(b)}{=} \max_p \frac{h(p)}{1 + p} = C,
\end{aligned}
\]
where $n=l + \lfloor lp \rfloor$, $(a)$ follows because using Stirling's approximation it can be shown that $\log {{l}\choose{k}} = l\cdot h(\frac{k}{l}) + o(l)$, $(b)$ holds since the binary entropy function $h(\cdot)$ is continuous, and the last equality follows from Theorem \ref{theorem111}.

\section{Proof of Theorem \ref{optcode2}}
\label{app:optcode2}

It suffices to show that the best rate is contained in the search range of $ l = \hat{l}$ to $2\hat{l}$. 

First, we note that if $l < \hat{l}$, then $2^l < M$, which implies that there is not a sufficient number of codewords in level $l$ to cover all messages. On the other hand, since $l < \frac{\eta(\mathcal{C}_l)}{M} < 2l$, we have
\[
\frac{\log M}{2l} < R_l < \frac{\log M}{l},
\]
where $R_l$ is the information rate of the optimum codebook at level $l$. Therefore, for all $l > 2\hat{l}$,
\[
R_l < \frac{\log M}{l} < \frac{\log M}{2\hat{l}} < R_{\hat{l}}.
\]
In other words, for all $l > 2\hat{l}$ the optimum information rate is less than the optimum information rate of $\mathcal{C}_{\hat{l}}$. This implies that there is no need to check any level lower than $\hat{l}$.


\section{Proof of Theorem \ref{theorem2}}
\label{AppendixB}

The proof consists of achievability and converse arguments.\\
\textbf{Converse:} For the CQC in a shared round robin scheduler with service rate 1 depicted in Figure \ref{SystemModelfig}, any code consisting of a codebook of $M$ equiprobable binary codewords, where messages take on average $n$ time slots to be received, satisfies
\begin{equation*}
\begin{aligned}
\displaystyle\frac{1}{n}\log{M}&\stackrel{(a)}{=}\frac{1}{n}H(W)\\
&=\frac{1}{n}I(W;\hat{W})+\frac{1}{n}H(W|\hat{W})\\
&\overset{(b)}{\le}\frac{1}{n}I(W;\hat{W})+\epsilon_{n}\\
&\overset{(c)}{\le}\frac{1}{n}I(X^m;Y^m)+\epsilon_{n},
\end{aligned}
\end{equation*}
where $(a)$ holds because $W$ is a uniform random variable over the message set $\{1,...,M\}$, $(b)$ follows from Fano's inequality with $\epsilon_{n}=\frac{1}{n}(H(P_e)+P_e\log_2{(M-1)})$, and $(c)$ follows from application of data processing inequality to the Markov chain in \eqref{MarkovChain}. 

Since the Z-channel model is memoryless,
\[
I(X^m;Y^m)\le \sum_{i=1}^mI(X_i;Y_i).
\]
Therefore,
\begin{equation*}
\label{eq0}
\begin{aligned}
\displaystyle\frac{1}{n}\log{M}
&\le\sum_{i=1}^m\frac{1}{n}I(X_i;Y_i)+\epsilon_{n}\le\max_{P_X}\frac{m}{n}I(X;Y)+\epsilon_{n}.
\end{aligned}
\end{equation*}
By Lemma \ref{lem:frac2}, in the scheme with the maximum information transmission rate between Alice and Bob, we have $\frac{m}{n}=\frac{1}{(1-p)+\delta p+2(1-\delta)p}$. Also, we note that
\begin{equation*}
\begin{aligned}
\label{eq2}
I(X;Y) = & h(Y) - h(Y|X) =  h((1 - \delta)p) - ph(\delta).
\end{aligned}
\end{equation*}
Therefore, we have
\begin{equation*}
\begin{aligned}
\displaystyle\frac{1}{n}\log{M}\le\max_{0\le p\le 1} \frac{h((1 - \delta)p) - ph(\delta)}{(1-p)+\delta p+2(1-\delta)p}+\epsilon_{n}.
\end{aligned}
\end{equation*}
As $n$ goes to infinity, $\epsilon_{n}\rightarrow 0$, and we have
\begin{equation*}
\begin{aligned}
C\le\max_{0\le p \le1} \frac{h((1 - \delta)p) - ph(\delta)}{(1-p)+\delta p+2(1-\delta)p}.
\end{aligned}
\end{equation*}

\noindent
\textbf{Achievability:}
Fix a Bernoulli distribution $P$ with parameter $p^*$, where
\begin{equation*}
\begin{aligned}
p^* = \arg \max_{0\le p \le1} \frac{h((1 - \delta)p) - ph(\delta)}{(1-p)+\delta p+2(1-\delta)p},
\end{aligned}
\end{equation*}
and generate a binary codebook $\mathcal{C}$ containing $2^{mR}$ i.i.d. sequences of  length $m$, drawn according to $P$, where $m=n/[(1-p^*)+\delta p^*+2(1-\delta)p^*]$.

In order to send a bit `1', Alice sends a packet and then idles for one time slot. To send a bit `0', she just idles for one time slot. Thus, each message on average takes $m[(1-p^*)+\delta p^*+2(1-\delta)p^*] = n$ time slots to be transmitted. At the same time, Bob keeps his head of the queue always full.

Since the Z-channel is a discrete memoryless channel, by the standard random decoding arguments \cite[Chapter 7]{cover2012elements}, the error can be kept arbitrary close to zero as $n$ goes to infinity as long as $R\le\max_{0\le p\le1} I(X;Y)$.
Consequently,
\begin{equation*}
\begin{aligned}
C & \ge\frac{\log 2^{m \times \max_{0\le p\le1} I(X;Y)}}{n} \\
& \ge \max_{0\le p\le1} \frac{h((1 - \delta)p) - ph(\delta)}{(1-p)+\delta p+2(1-\delta)p}.
\end{aligned}
\end{equation*}
The achievability and converse parts above, complete the proof of the coding theorem.

\end{appendices}

\section*{Acknowledgment}
This work was in part supported by NSF grant CCF 17-04970, and SaTC CORE 17-18952.

\bibliographystyle{ieeetr}
\bibliography{thesisrefs}

\begin{thebibliography}{10}

\bibitem{gong2011information}
X.~Gong, N.~Kiyavash, and P.~Venkitasubramaniam, ``Information theoretic
  analysis of side channel information leakage in {FCFS} schedulers,'' in {\em
  Information Theory Proceedings (ISIT), 2011 IEEE International Symposium on},
  pp.~1255--1259, IEEE, 2011.

\bibitem{gong2014quantifying}
X.~Gong and N.~Kiyavash, ``Quantifying the information leakage in timing side
  channels in deterministic work-conserving schedulers,'' {\em arXiv preprint
  arXiv:1403.1276}, 2014.

\bibitem{kadloor2010low}
S.~Kadloor, X.~Gong, N.~Kiyavash, T.~Tezcan, and N.~Borisov, ``Low-cost side
  channel remote traffic analysis attack in packet networks,'' in {\em
  Communications, 2010 IEEE International Conference on}, IEEE, 2010.

\bibitem{chen2015schedule}
C.-Y. Chen, A.~Ghassami, S.~Nagy, M.-K. Yoon, S.~Mohan, N.~Kiyavash, R.~B.
  Bobba, and R.~Pellizzoni, ``Schedule-based side-channel attack in
  fixed-priority real-time systems,'' tech. rep., 2015.

\bibitem{chen2017reconnaissance}
C.-Y. Chen, A.~Ghassami, S.~Mohan, N.~Kiyavash, R.~B. Bobba, R.~Pellizzoni, and
  M.-K. Yoon, ``A reconnaissance attack mechanism for fixed-priority real-time
  systems,'' {\em arXiv preprint arXiv:1705.02561}, 2017.

\bibitem{liberatore2006inferring}
M.~Liberatore and B.~N. Levine, ``Inferring the source of encrypted http
  connections,'' in {\em Proceedings of the 13th ACM conference on Computer and
  communications security}, pp.~255--263, ACM, 2006.

\bibitem{song2001timing}
D.~X. Song, D.~Wagner, and X.~Tian, ``Timing analysis of keystrokes and timing
  attacks on {SSH}.,'' in {\em USENIX Security Symposium}, vol.~2001, 2001.

\bibitem{wright2010uncovering}
C.~V. Wright, L.~Ballard, S.~E. Coull, F.~Monrose, and G.~M. Masson,
  ``Uncovering spoken phrases in encrypted voice over ip conversations,'' {\em
  ACM Transactions on Information and System Security (TISSEC)}, vol.~13,
  no.~4, p.~35, 2010.

\bibitem{tahir2015sneak}
R.~Tahir, M.~T. Khan, X.~Gong, A.~Ahmed, A.~Ghassami, H.~Kazmi, M.~Caesar,
  F.~Zaffar, and N.~Kiyavash, ``Sneak-peek: High speed covert channels in data
  center networks,'' in {\em IEEE International Conference on Computer
  Communications (INFOCOM)}, IEEE, 2016.

\bibitem{murdoch2005embedding}
S.~J. Murdoch and S.~Lewis, ``Embedding covert channels into {TCP/IP},'' in
  {\em International Workshop on Information Hiding}, 2005.

\bibitem{llamas2005evaluation}
D.~Llamas, A.~Miller, and C.~Allison, ``An evaluation framework for the
  analysis of covert channels in the {TCP/IP} protocol suite.,'' in {\em ECIW},
  pp.~205--214, 2005.

\bibitem{kang1996network}
M.~H. Kang, I.~S. Moskowitz, and D.~C. Lee, ``A network pump,'' {\em IEEE
  Transactions on Software Engineering}, vol.~22, no.~5, pp.~329--338, 1996.

\bibitem{anantharam1996bits}
V.~Anantharam and S.~Verdu, ``Bits through queues,'' {\em IEEE Transactions on
  Information Theory}, vol.~42, no.~1, pp.~4--18, 1996.

\bibitem{ghassami2018covert}
A.~Ghassami and N.~Kiyavash, ``A covert queueing channel in fcfs schedulers,''
  {\em IEEE Transactions on Information Forensics and Security}, vol.~13,
  no.~6, pp.~1551--1563, 2018.

\bibitem{soltani2015covert}
R.~Soltani, D.~Goeckel, D.~Towsley, and A.~Houmansadr, ``Covert communications
  on poisson packet channels,'' in {\em Communication, Control, and Computing
  (Allerton), 2015 53rd Annual Allerton Conference on}, pp.~1046--1052, IEEE,
  2015.

\bibitem{soltani2016covert}
R.~Soltani, D.~Goeckel, D.~Towsley, and A.~Houmansadr, ``Covert communications
  on renewal packet channels,'' in {\em Communication, Control, and Computing
  (Allerton), 2016 54th Annual Allerton Conference on}, pp.~548--555, IEEE,
  2016.

\bibitem{mukherjee2016covert}
P.~Mukherjee and S.~Ulukus, ``Covert bits through queues,'' in {\em
  Communications and Network Security (CNS), 2016 IEEE Conference on},
  pp.~626--630, IEEE, 2016.

\bibitem{cabuk2004ip}
S.~Cabuk, C.~E. Brodley, and C.~Shields, ``Ip covert timing channels: design
  and detection,'' in {\em Proceedings of the 11th ACM Conference on Computer
  and Communications Security}, pp.~178--187, ACM, 2004.

\bibitem{wsj2011}
D.~Wakabayashi, ``Breach complicates {Sony}'s network ambitions,'' {\em The
  Wall Street Journal}, April 28, 2011.

\bibitem{gianvecchio2007detecting}
S.~Gianvecchio and H.~Wang, ``Detecting covert timing channels: an
  entropy-based approach,'' in {\em Proceedings of the 14th ACM conference on
  Computer and communications security}, pp.~307--316, ACM, 2007.

\bibitem{ghassami2015capacity}
A.~Ghassami, X.~Gong, and N.~Kiyavash, ``Capacity limit of queueing timing
  channel in shared {FCFS} schedulers,'' in {\em 2015 IEEE International
  Symposium on Information Theory (ISIT)}, pp.~789--793, IEEE, 2015.

\bibitem{kadloor2015delay}
S.~Kadloor and N.~Kiyavash, ``Delay-privacy tradeoff in the design of
  scheduling policies,'' {\em IEEE Transactions on Information Theory},
  vol.~61, no.~5, pp.~2557--2573, 2015.

\bibitem{srikant2013communication}
R.~Srikant and L.~Ying, {\em Communication networks: an optimization, control,
  and stochastic networks perspective}.
\newblock Cambridge University Press, 2013.

\bibitem{xie2016scheduling}
Q.~Xie, A.~Yekkehkhany, and Y.~Lu, ``Scheduling with multi-level data locality:
  Throughput and heavy-traffic optimality,'' in {\em INFOCOM 2016-The 35th
  Annual IEEE International Conference on Computer Communications, IEEE},
  pp.~1--9, IEEE, 2016.

\bibitem{cover2012elements}
T.~M. Cover and J.~A. Thomas, {\em Elements of Information Theory}.
\newblock John Wiley \& Sons, 2012.

\bibitem{csiszar2011information}
I.~Csiszar and J.~K{\"o}rner, {\em Information Theory: Coding Theorems for
  Discrete Memoryless Systems}.
\newblock Cambridge University Press, 2011.

\bibitem{yekkehkhany2018gb}
A.~Yekkehkhany, A.~Hojjati, and M.~H. Hajiesmaili, ``Gb-pandas:: Throughput and
  heavy-traffic optimality analysis for affinity scheduling,'' {\em ACM
  SIGMETRICS Performance Evaluation Review}, vol.~45, no.~2, pp.~2--14, 2018.

\bibitem{yekkehkhany2017near}
A.~Yekkehkhany, ``Near data scheduling for data centers with multi levels of
  data locality,'' {\em (Dissertation, University of Illinois at
  Urbana-Champaign)}.

\end{thebibliography}


%








\end{document}